\documentclass[]{article}
\usepackage[margin=1in]{geometry}
\usepackage{authblk}
\newcommand{\arxiv}[1]{#1}
\newcommand{\journal}[1]{}

\journal{\documentclass[english,preprint,JIP]{ipsj}
\bibliographystyle{ipsjunsrt-e}
}

\usepackage{graphicx}
\usepackage{latexsym}

\usepackage[square,sort,comma,numbers]{natbib}

\usepackage{amssymb,amsmath,amsthm}

\usepackage{hyperref}
\usepackage{enumerate}

{\makeatletter
 \gdef\xxxmark{%
   \expandafter\ifx\csname @mpargs\endcsname\relax %
     \expandafter\ifx\csname @captype\endcsname\relax %
       \marginpar{xxx}%
     \else
       xxx %
     \fi
   \else
     xxx %
   \fi}
 \gdef\xxx{\@ifnextchar[\xxx@lab\xxx@nolab}
 \long\gdef\xxx@lab[#1]#2{\textbf{[\xxxmark #2 ---{\sc #1}]}}
 \long\gdef\xxx@nolab#1{\textbf{[\xxxmark #1]}}
}

\newcommand{\defn}[1]{\textbf{\textit{\boldmath #1}}}

\let\epsilon=\varepsilon

\usepackage{xcolor}
\usepackage{soul}

\newcommand{\ie}{\textrm{i.e.,~}}

\def\Underline{\setbox0\hbox\bgroup\let\\\endUnderline}
\def\endUnderline{\vphantom{y}\egroup\smash{\underline{\box0}}\\}
\def\|{\verb|}

\journal{
\setcounter{volume}{26}%
\setcounter{number}{1}%
\setcounter{page}{1}

\received{2016}{3}{4}
\accepted{2016}{8}{1}
}

\newtheorem{theorem}{Theorem}
\newtheorem{problem}{Problem}
\newtheorem{lemma}{Lemma}

\usepackage[varg]{txfonts}%
\makeatletter%
\input{ot1txtt.fd}
\makeatother%

\begin{document}

\title{Continuous Flattening and Reversing \\ of Convex Polyhedral Linkages}

\journal{
\affiliate{mit}{MIT Computer Science and Artificial Intelligence Laboratory,
      32 Vassar St., Cambridge, MA 02139, USA}
\affiliate{meiji}{Meiji Institute for Advanced Study of Mathematical Sciences, Meiji University, Nakano, Tokyo 164-8525, Japan}
\author{Erik D. Demaine}{mit}[edemaine@mit.edu]
\author{Martin L. Demaine}{mit}[mdemaine@mit.edu]
\author{Markus Hecher}{mit}[hecher@mit.edu]
\author{Rebecca Lin}{mit}[ryelin@mit.edu]
\author{Victor Luo}{mit}[vluo@mit.edu]
\author{Chie Nara}{meiji}
[cnara@jeans.ocn.ne.jp]
}

\arxiv{%
\author[a]{Erik D. Demaine}
\author[a]{Martin L. Demaine}
\author[a]{Markus Hecher}
\author[a]{Rebecca Lin}
\author[a]{Victor H. Luo}
\author[b]{Chie Nara}
\affil[a]{MIT Computer Science and Artificial Intelligence Laboratory,
      32 Vassar St., Cambridge, MA 02139, USA, {\{edemaine,mdemaine,hecher,ryelin,vluo\}@mit.edu}}
\affil[b]{Meiji Institute for Advanced Study of Mathematical Sciences, Meiji University, Nakano, Tokyo 164-8525, Japan,
{cnara@jeans.ocn.ne.jp}
}%
\maketitle}

\begin{abstract}
We prove two results about transforming any convex polyhedron,
modeled as a linkage $L$ of its edges.
First, if we subdivide each edge of $L$ in half, then
$L$ can be continuously \emph{flattened} into a plane.
Second, if $L$ is equilateral and we again subdivide each edge in half,
then $L$ can be \emph{reversed}, i.e., turned inside-out.
A linear number of subdivisions is optimal up to constant factors, as we show
(nonequilateral) examples that require a linear number of subdivisions.
For nonequilateral linkages, we show that more subdivisions can be required:
even a tetrahedron can require an arbitrary number of subdivisions to reverse.
For nonequilateral tetrahedra, we provide an algorithm that matches this
lower bound up to constant factors: logarithmic in the aspect ratio.
\end{abstract}

\journal{
\maketitle}

\section{Introduction}

Given a polyhedral surface (planar polygons glued together in 3D),
what shapes can we fold it into via a continuous motion
that avoids crossings and stretching?
A first goal is \defn{continuous flattening}
\cite[Section 18.1]{demaine2007geometric}:
put the entire surface in a single plane.
Continuous flattening is known to be possible for all
convex polyhedra \cite{AbelEtAl14,ItohEtAl2012},
orthogonal polyhedra \cite{DemaineEtAl15},
nonconvex polyhedra with countably many creases \cite{AbelEtAl21},
and various special cases with additional properties
\cite{HoriyamaEtAl15,MatsubaraNara20}.
All of these results for closed surfaces
use ``rolling'' creases which move over time,
because otherwise the Bellows Theorem prevents changing the volume.
A second goal is \defn{reversing} or \defn{turning inside-out} \cite{Maehara10}:
transforming the surface into its mirror image,
exchanging the identity of the two sides (in the case of orientable surfaces).
This is known to be possible for a family of polyhedral surfaces with boundary
based on tubes, notably using just finitely many creases \cite{Maehara10}.

In this paper, we consider continuous flattening and turning inside-out
applied to just the \emph{edges} of a polyhedral surface.
More precisely, a \defn{linkage} is an assembly of rigid segments (edges)
whose lengths must be preserved,
connected together by universal joints (vertices).
In particular, we focus on \defn{convex polyhedral linkages},
where the edges form the edge skeleton of a convex polyhedron.
Like polyhedral surfaces, we allow the addition of extra creases,
or in other words, \defn{subdivision} of the edges.
Continuous flattening and turning inside-out of linkages is easy
with rolling creases, so we allow only finitely many creases,
i.e., finite subdivision.
We also forbid proper crossing between edges
(as detailed in Section~\ref{sec:model}),
in contrast to some past work on turning polygons inside-out
\cite[Section 5.1.2]{demaine2007geometric}.

More precisely, we study on the following two problems
for convex polyhedral linkages:

\smallskip
\begin{problem}[Flattening]\label{prob:flatten}
  Given a linkage $L$ with initial state/configuration $S$,
  a \defn{flattening} is a continuous motion of $L$ from $S$
  to a \defn{flat state} where all vertices lie in the $xy$ plane,
  while preserving edge lengths and avoiding crossings.
  Which polyhedral linkages can be continuously flattened?

  To \defn{subdivide} an edge $AB$ in a linkage is to replace $AB$ with a new vertex, say $X$, and edges $AX$ and $XB$ such that $|AX|+|XB|=|AB|$.
  How many subdivisions do we need to continuously flatten a given linkage?
  Does it suffice to subdivide each edge $k$ times for some~$k$?
\end{problem}

\smallskip
\begin{problem}[Turning Inside-Out/Reversing]
\label{prob:inside-out}
  A \defn{reversing} or \defn{turning inside-out} motion of a linkage $L$
  with labeled vertices and initial state $S$
  is a continuous motion of $L$
  from $S$ to its reflection $S'$ through a plane,
  while preserving edge lengths and avoiding crossings.
  Which linkages admit such a motion?
  How many subdivisions do we need to continuously flatten a given linkage?
  Does it suffice to subdivide each edge $k$ times for some~$k$?
\end{problem}
\smallskip

More generally, we could ask when there are motions between \emph{every}
two states of the linkage, i.e., connectivity of the configuration space.
Such problems have been considered extensively
\cite[Chapter~6]{demaine2007geometric}, but not with subdivision.

\subsection{Our Results}
\noindent
We prove the following results:
\begin{enumerate}
\item A worst-case linear lower bound on the number of required subdivisions
  for flattening or turning inside-out a (nonequilateral) linkage
  (Section~\ref{sec:lower}).
\item Every convex polyhedral linkage can be continuously flattened if we
  subdivide each edge at its midpoint (Section~\ref{sc:s-ball}).
  This bound matches the lower bound up to constant factors.
\item Every equilateral convex polyhedral linkage can be turned inside-out
  if we subdivide each edge at its midpoint (Section~\ref{sc:i-o}).
\item Turning a (nonequilateral) tetrahedron inside-out requires a number
  of subdivisions that is logarithmic in its aspect ratio
  (matching upper and lower bounds).
  More generally, we establish a lower bound for all polyhedral linkages.
\end{enumerate}

\section{Definitions}\label{sec:model}

A \defn{linkage} \cite[Chapter 11]{demaine2007geometric}
is a graph paired with an assignment of lengths to its edges.
Each edge represents a rigid bar, and
each vertex acts as a joint about which incident edges can freely rotate.
A linkage is \defn{equilateral} if all its edge lengths are equal.
A \defn{[convex] polyhedral linkage} is a linkage obtained by taking
the edge skeleton of a [convex] polyhedron.

A \defn{state} or \defn{configuration} of a linkage specifies
a straight-line drawing of the graph in Euclidean space
where the lengths of the embedded segments match the assigned edge lengths.
A state is \defn{nontouching} if the drawing is in fact an embedding,
i.e., no edges touch except at shared endpoints.
We sometimes also allow a state to be \defn{self-touching},
meaning that edges can intersect; in this case, a state also specifies
the local connectivity in an infinitesimal neighborhood of each point
where multiple edges meet (including vertices and touching points)
as a topological drawing of the neighborhood within an infinitesimal ball;
see \cite{ConnellyDemaineRote02} for details.
Figure~\ref{fig:cross}~shows an example of two self-touching states of a
linkage with identical vertex positions but distinct edge orderings.
A \defn{motion} of a linkage is a continuum $\{S_t \mid t \in [0,1]\}$
of states that avoids \defn{crossings},
i.e., limits of states that approach two edges touching
match the topological drawings once they touch,
and if touching remains for a positive time,
the topological drawings are preserved up to homotopy.
\begin{figure}[htbp]
\centering
\includegraphics[width=\arxiv{.75}\columnwidth]
{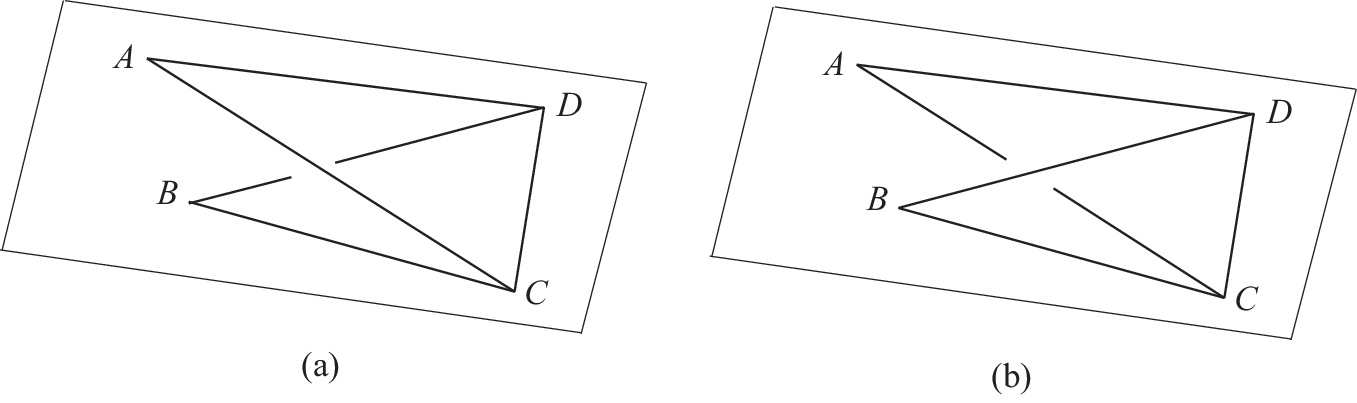}
\caption{Two distinct flat-folded states of a linkage where edges $AC$ and $BD$ are touching. In (a), $AC$ lies above $BD$ (\ie $AC \succ BD$), while in (b), $BD$ lies above $AC$ (\ie $BD \succ AC$). Transitioning between these states requires a $2\pi$ rotation about edge $CD$.}
\label{fig:cross}
\end{figure}

A simpler model for touching-but-not-crossing linkages is
\defn{sticky linkages}: once two edges touch,
they remain touching at the same points
(measured relative to their lengths).
Our motions will all follow this stronger model,
which clearly forbids crossing.

We are interested in \defn{continuously flattening}
a convex polyhedral linkage $L$, i.e., folding the linkage $L$ via a motion
that ends with a \defn{flat} state that lies entirely within a plane.
Notably, a continuous flattening $\{S_t \mid t \in [0, 1]\}$ must satisfy
the following properties:
\begin{enumerate}
\item $S_0$ is the given initial state of the linkage;
\item $S_t$ remains isomorphic to $L$, and all edges maintain their original lengths;
\item $S_1$ is flat, meaning that all its vertices lie in a plane; and
\item edges may touch but must not cross (e.g., sticky model).
\end{enumerate}

\section{Lower Bound on Subdivision}
\label{sec:lower}

\noindent
First we show that,
in the worst case, we need a linear number of edge subdivisions to flatten or
turn inside-out nonequilateral polyhedra:

\smallskip
\begin{theorem}
  There is an infinite family of convex polyhedral linkages that requires
  $\Omega(|E(G)|)$ subdivisions to flatten or reverse,
  where $G$ is the graph of the linkage.
\end{theorem}
\begin{proof}
    Start from a family of strictly convex triangulated convex polyhedra~$P$,
    such as geodesic domes.
    Subdivide each triangle $T$ of $P$ into three triangles
    by adding a central vertex $v_T$ slightly off the triangle~$T$
    (small enough to preserve strict convexity)
    and connecting $v$ to the three vertices of the triangle~$T$.
    If $P$ has $n$ triangles, then $P'$ has $3n$ triangles,
    dividing into $n$ triples.
    Each triple of triangles forms a tetrahedron, which is rigid and not flat
    by construction, so at least one of its edges must be subdivided in order
    for it to either flatten or turn inside-out.
    Some edges are shared by two tetrahedra, but only two,
    so we need at least $n/2$ subdivisions.
    By the Handshaking Lemma, $P$ has $3n/2$ edges,
    and $P'$ has $|E(G)| = 9n/2$ edges.
    Thus we need $|E(G)|/9 = \Omega(|E(G)|)$ subdivisions.
\end{proof}

\section{Flattening Convex Polyhedral Linkages} \label{sc:s-ball}

\noindent We provide an algorithm for flattening that matches the linear lower bound above.
More precisely, we establish the following.

\smallskip
\begin{theorem} \label{th:convex}
Given a convex polyhedral linkage, subdividing every edge at its midpoint enables continuous flattening.
Moreover, the resulting linkage can be continuously moved to lie in a straight line.
\end{theorem}
\begin{proof}[Proof]
Let $P$ be the convex hull of a convex polyhedral linkage $L$.  Figure \ref{fig:2pf} shows An example of continuously folding a convex polygonal linkage into a spiky ball by subdividing all edges in half. The algorithm is as follows.

\begin{figure}[htbp]
\centering
\includegraphics[width=\arxiv{.75}\columnwidth]
{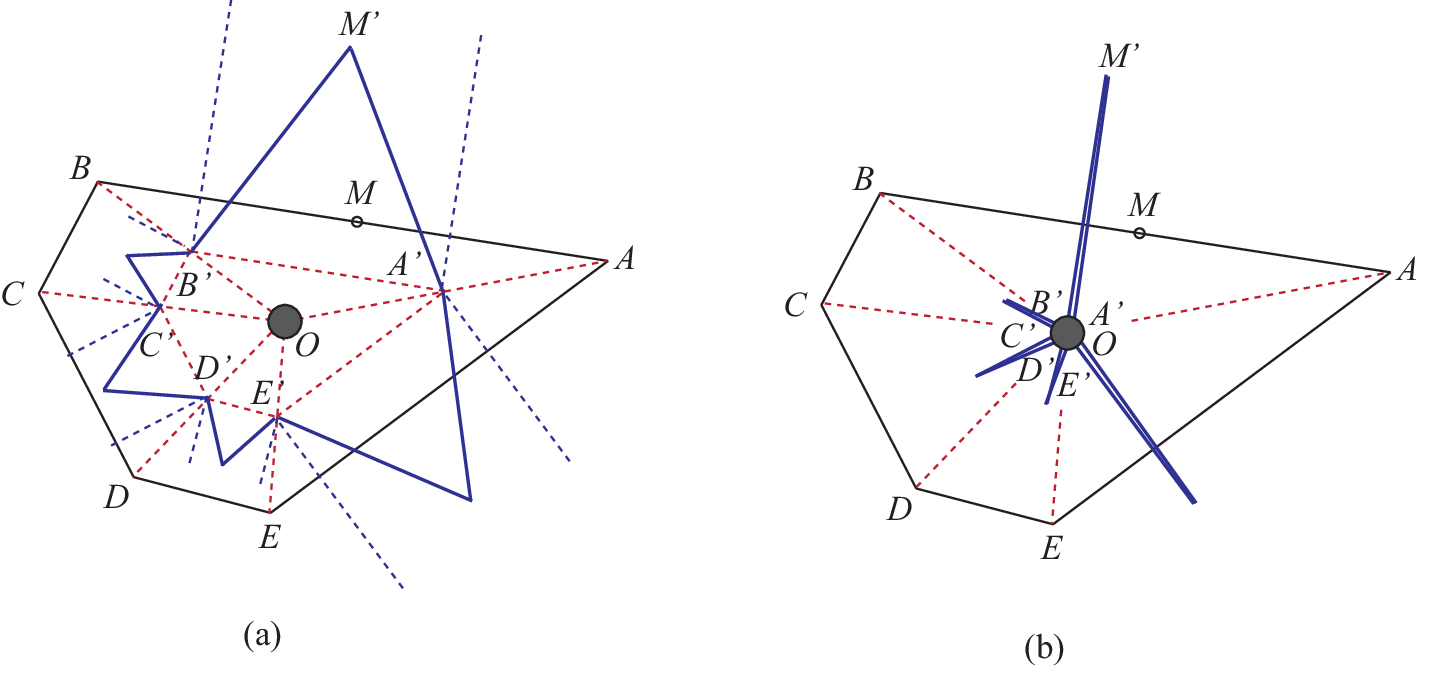}
\caption{An example of continuously folding a convex polygonal linkage into a spiky ball by subdividing all edges in half: (a) the original pentagonal linkage $ABCDE$ with a fixed point $O$, the midpoint $M$ of the edge $AB$, and its intermediate folding with vertices $A', B', C', D', E'$ and $M'$ corresponding to $A, B, C, D, E$ and $M$, respectively; and (b) the spiky ball with center $O$ reached by $A', B', C', D'$ and $E'$.}
\label{fig:2pf}
\end{figure}

\begin{enumerate}
\item Choose any point $O$ in the interior of $P$ or the interior of any face.
\item Shrink $L$ toward $O$ by moving all vertices of $P$ toward $O$ and folding all edges in half at their midpoints such that, at each moment, the convex figure $P'$ of moved vertices of $P$ is similar to $P$, and each edge $e$ of $P$ is folded in half and sticks outside of the triangle composed of $e$ and $O$ in the plane $OAB$. Since $P$ is convex and the folded figure of $e$ by the midpoint comprises an isosceles triangle, there are no collisions in the above motion. Indeed, each edge is between two dashed blue orthogonal slabs, as drawn in Figure~\ref{fig:2pf} (a), which are disjoint and resulting triangles remain inside (see also Figure~\ref{fig:2pf} (b)).
\item When all vertices reach $O$, we get a {\em spiky ball} that can be flattened and continuously moved to a straight line.
This works as follows. First, we translate to make the center at $O$. Then, we iteratively take the spike with the smallest angle to the target axial ($xy$) plane and directly rotate it to the plane. Since this was the spike with the smallest angle, it does not hit any other spikes, except possibly touching the plane. We repeat until reaching a folded state.
We can move to a line by using the same method.
\end{enumerate}~\\[-3.75em]
\end{proof}

We call the point $O$ in the above proof the \defn{center} of the spiky ball. Figure \ref{fig:linf} (a) shows a continuous process for a tetrahedral linkage flattened via a spiky ball. Figure \ref{fig:linf} (b) shows another continuous flattening process to obtain Figure \ref{fig:cross} (a) as a subset, that is more efficient in that it requires only a single subdivision, but it cannot be transformed to a straight line without additional subdivision of edges.

\begin{figure}[htbp]
\centering
\includegraphics[width=\arxiv{.75}\columnwidth]{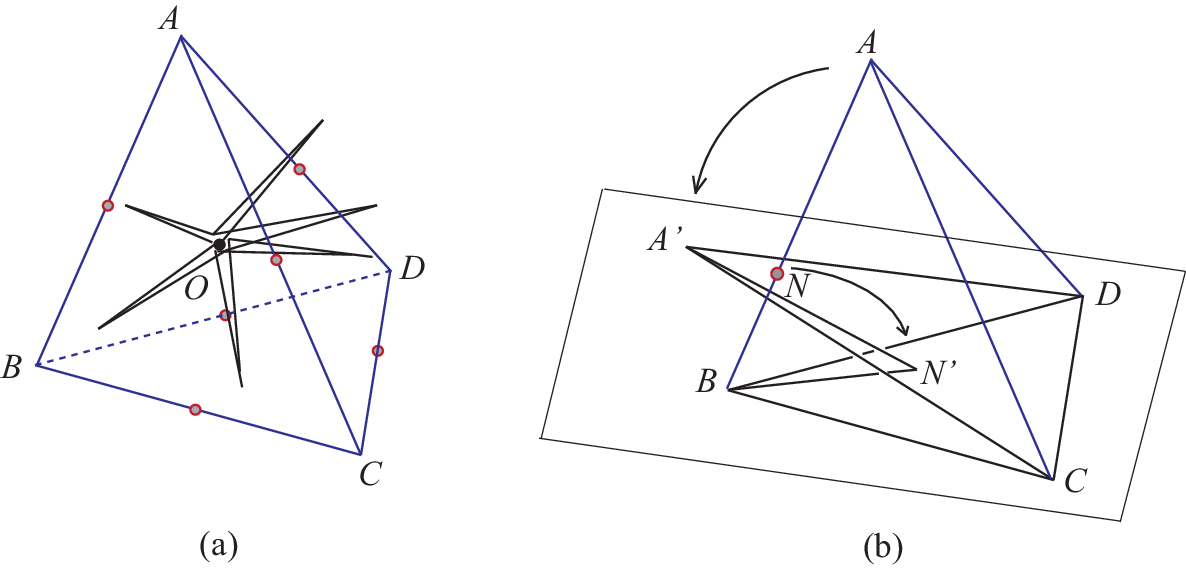}
\arxiv{.75}\caption{An example of flattening a tetrahedral linkage in two ways: (a) subdividing all edges in half, and (b) subdividing only one edge to obtain Figure \ref{fig:cross} (a) as a subset.}
\label{fig:linf}
\end{figure}

\section{Turning Equilateral Linkages Inside-Out} \label{sc:i-o}
\noindent In this section, we work on the problem of continuous inside-out (reversing), that is, we will establish the following result. %

\smallskip
\begin{theorem} \label{th:i-o}
Given an equilateral convex polyhedral linkage, subdividing every edge at its midpoint enables continuous reversing, that is, inside-out.
\end{theorem}
\smallskip

Before establishing an algorithm that is crucial for this proof, we focus on two special cases that will be key in the algorithm.

\subsection{Tetrahedral Linkages}

\noindent The following lemma plays a key role in the algorithm and is therefore essential for the proof of Theorem \ref{th:i-o}.

\smallskip
\begin{lemma}\label{lem:pyramid}
Given a tetrahedral linkage $L$ with an equilateral triangular base $BCD$ and the peak $A$ such that $|AB|=|AC|=|AD| < 2 |BC|/\sqrt3$, subdividing every edge, connected to $A$, at its midpoint enables continuous reversing, that is, turning it inside-out.
\end{lemma}

\begin{proof}
Let $L$ be the linkage given in the above and $O$ be the orthogonal foot of  $A$ onto the triangle $BCD$  (see Figure~\ref{fig:py-i-o} (a)). Fix the triangle $BCD$. Move the peak $A$ to its mirror image $A'$ about the triangle $BCD$ along the line segment $AA'$. Simultaneously, rotate the midpoints of the edges $AB, AC$ and $AD$ about $B, C$ and $D$, respectively, along circular arcs of radius $|AB|/2=|AC|/2=|AD|/2$ in the planes determined by $O$ and $AB, AC$ and $AD$, respectively. No collisions occur during the motion, because the conditions
$|AB|=|AC|=|AD| < 2 |BC|/\sqrt3$ and $|AO|=|BO| =|CO|= |BC|/\sqrt3$ yield $|BO| > |AB|/2$, $|CO| > |AC|/2$, and $|DO| > |AD|/2$. Note that the traces of midpoints of $B, C$, and $D$
 pass through their destinations which are the mirror images of midpoints about the triangle $BCD$, and then come back there (see Figure~~\ref{fig:py-i-o} (b)--(d)).
\end{proof}

\begin{figure}[htbp]
\centering
\includegraphics[width=\arxiv{.75}\columnwidth]{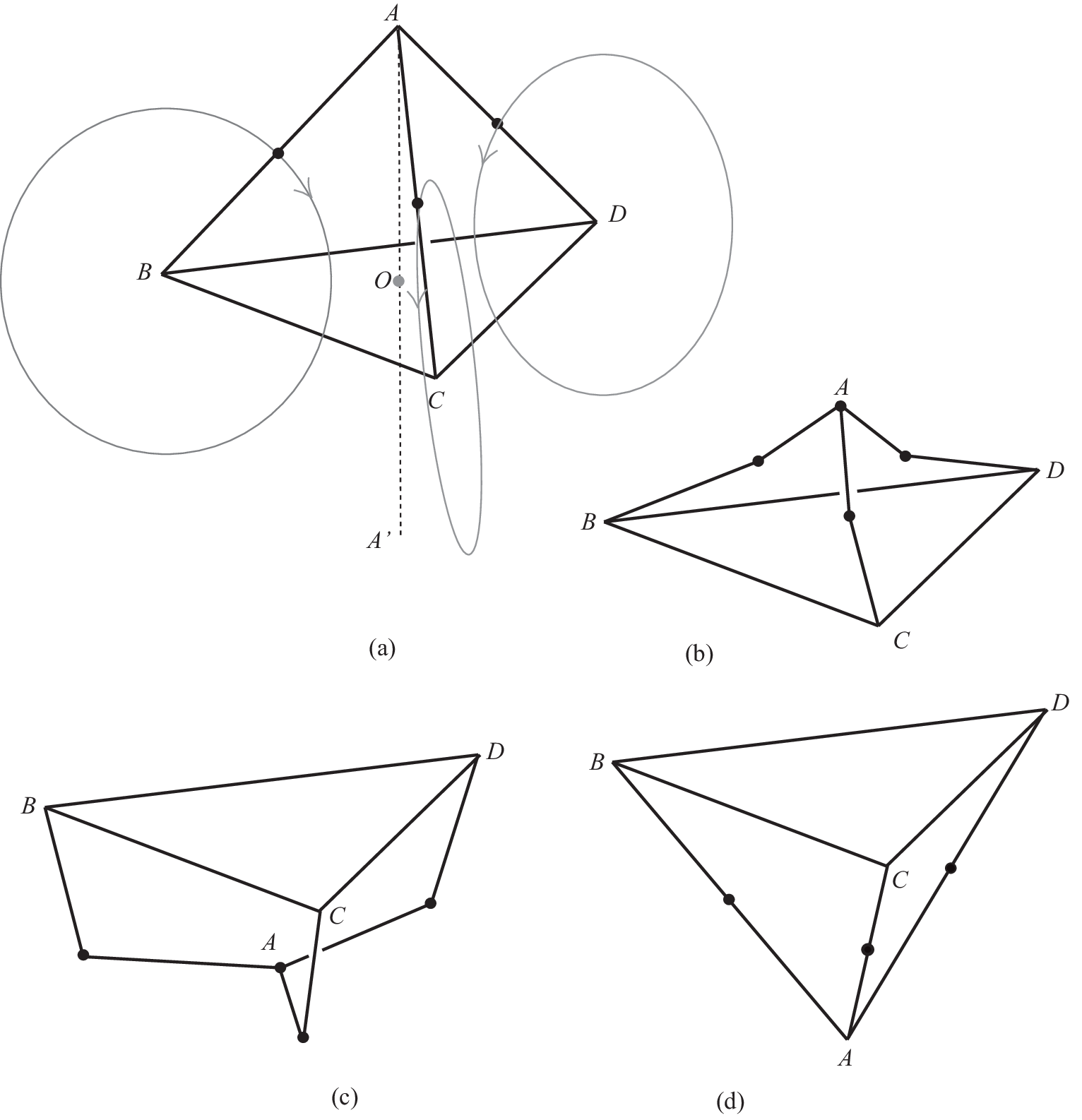}
\caption{(a) An example of a tetrahedral linkage with an equilateral base $BCD$ and the peak $A$ such that $|AB|=|AC|=|AD| < 2 |BC|/\sqrt3$; (b) Motions of $A$ and midpoints of $AB, AC, AD$ in the beginning; (c) The figure of the linkage when midpoints reach on the cylinder of radius $|AB|/2$ with the $AA'$-axis (where $A'$ is the target location of $A$); (d) The final figure of $L$ which is the mirror image of the original.
\arxiv{.75}}
\label{fig:py-i-o}
\end{figure}

\noindent {\bf Remark}. In the above lemma, the continuous motion of inside-out can be applied to more general tetrahedron $ABCD$ satisfying the following two conditions:

(1) the orthogonal foot $O$ of $A$ on the plane including the triangle $BCD$ is in the interior of the triangle,

(2) $|AB|/2 < |OB|, | AC|/2 <2|OC|, \rm{and} \, |AD|/2 < |OD|$.

\subsection{Triangular Prisms}

\noindent Now, with this result at hand, we show a method of turning inside-out a stack of triangular prisms,
which is the core of what we do in the general case.

\smallskip
\begin{lemma}\label{lem:prism}
Given the linkage consisting of a stack of $k$ equilateral triangles connected by $k-1$ prisms (not necessarily equilateral), totaling $3k$ vertices, as shown in Figure~~\ref{fig:prism-i-o} (a). Then, subdividing every edge except the base $v_1v_2v_3$ at its midpoint enables continuous reversing, that is, turning it inside-out.
\end{lemma}

\begin{figure}[htbp]
\centering
\includegraphics[width=\arxiv{.75}\columnwidth]{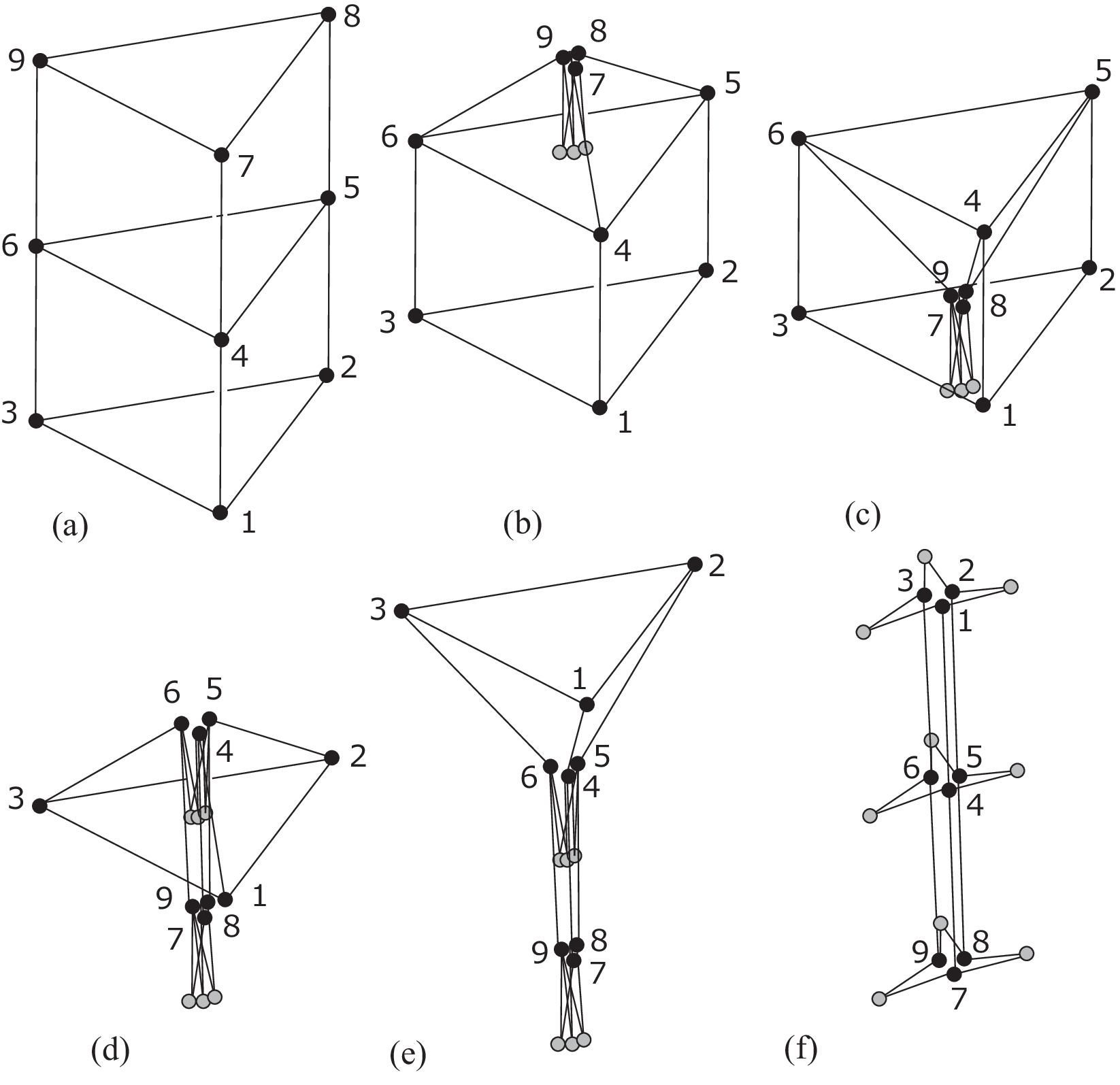}
\caption{(a) A linkage of an equilateral triangular prism with three stacks, where number $i$ stands for the vertex $v_i$; (b) The top stack transformed to a near tetrahedral linkage; (c) The near tetrahedral linkage transformed to the inside-out figure: (d) The second top stack transformed to a near tetrahedral linkage; (e) The near tetrahedral linkage transformed to the inside-out figure; (f) Edges folded in half moved in the outside.
}
\label{fig:prism-i-o}
\end{figure}

\begin{proof} The algorithm of the proof is as follows.
\begin{enumerate}
\item  Denote by $\gamma$ the central axis of $L$. Move three vertices $v_i \, (i=7,8,9)$ of the top stack toward $\gamma$ by rotation about $v_i \, (i=4,5,6)$, respectively until just before touching each other. Simultaneously fold three edges $v_7v_8, v_8v_9$ and $v_9v_7$ at their midpoints, hanging down from their end vertices. We call the resulting figure %
a {\em near-tetrahedral} linkage, which is attached by folded edges hanging down, as shown in Figure~\ref{fig:prism-i-o} (b).
\item Apply the motion ``inside-out" to the near-tetrahedral linkage, where edges hanging down move together with their end vertices as they are (Figure~\ref{fig:prism-i-o} (c)).
\item Apply a similar process to the second top stack. Then, we get a reversed near-tetrahedral linkage with edges hanging down from their end vertices (Figure~\ref{fig:prism-i-o} (d) and (e)).
\item Move folded edges at midpoints toward the outside of the cylinder until reaching the plane orthogonal to $\gamma$ (Figure~\ref{fig:prism-i-o} (f)).
\item Stretching folded edges, keeping vertical edges straight, enables turning the original figure of $L$ inside-out.%
\end{enumerate}~\\[-3.5em]%
\end{proof}

\smallskip
\noindent {\bf Remark}. In the above lemma, the continuous motion of inside-out can be applied to a more general triangular prism with an equilateral triangular base of edge length $l$  such that the height of each stack is greater than $(1/\sqrt3) l$ and less than $(2/\sqrt3) l$.

\subsection{General Algorithm}

\noindent The results from above can now be used to develop even more generic procedures for turning inside-out. The outline of our algorithm for turning equilateral linkages inside-out is as follows. Here and in what follows, we assume the edge length of a given linkage is one.

\begin{figure}
    \centering
    \includegraphics[width=\arxiv{.75}\linewidth]{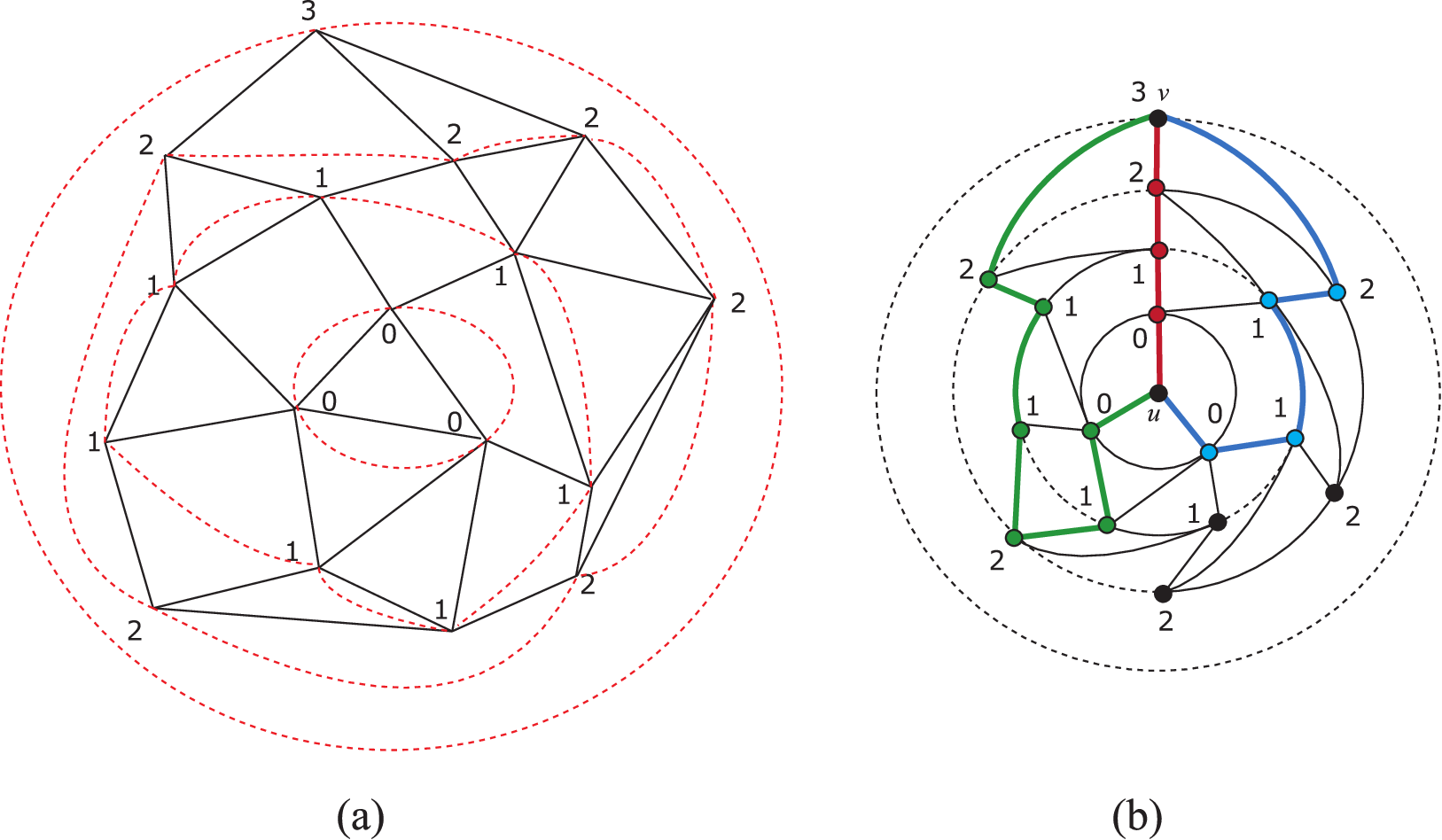}
    \caption{(a): Visualization of vertex sets $V_i$ according to their distance to $F$ obtained by breath-first search. (b): Horizontal edges are along the circle and vertical edges connect neighboring circles. By 3-connectivity there exist $3$ disjoint paths, as highlighted in color.}
    \label{fig:bfs}
\end{figure}

\begin{enumerate}
\item \emph{Choose any face $F$ of $L$ and any interior point $O$ of $F$}. Draw Euclidean $xyz$-axes with the origin $O$ such that $F$ is on the $xy$-plane and $L$ in the halfspace $\{z \ge 0 \}$.
\item \emph{Breadth-first search}. Divide the vertex set $V$ into subsets $\{V_i: 0 \le i \le \ell\}$ according to the edge distance from $F$.
We call an edge whose endpoints are in the same subset {\it horizontal edge} and any other edge \emph{vertical edge}. %
Figure~\ref{fig:bfs} (a) depicts vertex sets $V_i$ according to their distance $i$ to $F$; (b) shows the location of horizontal edges along the circles and vertical edges connecting neighboring circles.

\item \label{step:spiky} \emph{Transform $L$ into a near spiky ball, denoted by $nsb(L)$}, by stopping the continuous motion $\{L_t: 0 \le t \le 1\}$ as defined in the proof of Theorem \ref{th:convex}, just before the end.

\item \emph{Subdivide each $V_i$ into three groups $V_{i,1}$, $V_{i,2}$, and $V_{i,3}$, evenly in equator order}, where $V_{\ell,1}=V_{\ell}$, and $V_{\ell,2}$ and $V_{\ell,3}$ are empty in the case $|V_{\ell}| \le 2$.
We do so by exploiting 3-connectivity, implying the existence of three disjoint paths by Menger's Theorem~\cite{Menger27,Goring00}.
Note that the in-between ``vertical'' edges form a planar graph, see also Figure~\ref{fig:bfs} (b).
For details, we refer to the proof in Section \ref{pr:i-o} below. %

\item \label{step:stretch} Move vertices of {\it nsb(L)} by parallel transformation along $z$-axis to a thin cylindrical $({\ell}+1)$-sized stack, such that the vertices of $V_i$ are on the plane $z=i \,(1-\epsilon)$ with small enough $\epsilon >0$.
This works in multiple steps: First the maximum distance vertices in $V_{\ell}$ move up by $1-\epsilon$, then the farthest two layers $V_{{\ell}-1},V_{\ell}$ together move up by $1-\epsilon$, and so on.
Edges follow their end vertices, where their midpoints are located in the exterior of the cylinder; horizontal edges are located on the planes orthogonal to the $z$-axis. We call the resulting figure a {\it cylindrical linkage}, denoted by $cy(L)$.

\item  \emph{Transform $cy(L)$ to a near triangular prism with extra edges}, e.g., Figure~\ref{fig:icosa_i-o} (d) for a regular icosahedral linkage.
 To do so, move vertices so that the vertices in the same subgroup are located very close to each other, and $V_{i,1}$, $V_{i,2}$, and $V_{i,3}$ make a figure close to an equilateral triangle.
Denote by $ntp(L)$ the resulting figure of $L$.
Note that the height of $ntp(L)$ is smaller in this step (by some constant factor) so that the vertical edges can reach and still be at most unit length.

\item \emph{Transform the top part of $ntp(L)$ to a modified tetrahedral linkage with extra edges}. See Figure~\ref{fig:icosa_i-o_2} (b),
which is obtained from Figure~\ref{fig:icosa_i-o_2} (a) via a near tetrahedron, where midpoint triangles are rotated so they are inside, thereby including both horizontal and vertical edges.
Then, apply Lemma \ref{lem:pyramid} to the top part,
with some modifications to handle additional ``diagonal'' vertical edges
to also go down with the other edges.

\item \emph{Continue the above process until all stacks are turned inside-out.} Then shrink the top inverted tetrahedron: move the top three vertices toward the center of the cylinder, moving midpoints toward the outside and in the horizontal plane. The resulting figure is a very thin cylinder $cy(L)$. Move all horizontal edges to the exterior of the cylinder.

\item \emph
{Transform the resulting figure to the reflection of $nsb(L)$
(opposite of Step~\ref{step:stretch}), and then transform it to a convex polyhedron (opposite of Step~\ref{step:spiky})}, which is a mirror image of $L$ by the uniqueness of realization of a convex polyhedron up to congruence.
\end{enumerate}

\subsection{Example: Regular Icosahedron}
\noindent Before we discuss the precise proof of Theorem \ref{th:i-o}, we illustrate the algorithm based on the regular icosahedron (see Figures \ref{fig:icosa_i-o}. \ref{fig:icosa_i-o_1}. and \ref{fig:icosa_i-o_2}).
Let $L$ be the linkage of a regular icosahedron with vertex set $V=\{v_i: i=1,2,\dots, 20\} $ as shown in Figure~\ref{fig:icosa_i-o} (a) and (b).
\medskip

{\bf Steps 1 and 2.}  Choose any face $F$, say $F=[v_1v_2v_3]$, and the center of $F$ as the center $O$ of a spiky ball of $L$. We assume the edge length is one and $L$ is depicted in $\mathbb{R}^3$ with the origin $O$ and $xyz$-axes so that $F$ is on the $xy$-plane and $L$ in the halfspace $z \ge 0$.

{\bf Step 3.}  Apply Theorem \ref{th:convex} to $L$.  Stop the continuous motion of the spiky ball just before the end. Then we get a figure close to the spiky ball (Figure~\ref{fig:icosa_i-o} (c)), denoted by $nsb(L)$.  We will transform $nsb(L)$ to a kind of triangular prism as shown in Figure~\ref{fig:icosa_i-o} (d) in the next step.

{\bf Step 4.} Divide $V$ into groups according to the edge-distance from $F$ and denote them by
$$V_0=\{v_i: i=1, 2, 3\}, V_1=\{v_i: i=4, 5, 6, 7, 8, 9 \}, \rm{and} $$
$$V_2=\{v_i: i=10, 11, 12\}.$$
 Transform $nsb(L)$ to a thin cylindrical 3-sized stack, denoted by $cy(L)$, by moving the  vertices in $V_1$  and $V_2$ to the plane $\{z=1-\epsilon\}$ and the plane $\{z=2-2\epsilon\}$, respectively by translation parallel to the  $z$-axis, where $\epsilon$ is a enough small positive number (Figure~\ref{fig:icosa_i-o_1} (b).

 Each horizontal edge is located in the plane parallel to the $xy$-plane and each vertical edge folded at its midpoint is located outside so that its angle-bisector at the midpoint meets the $z$-axis.

Next, we will continuously transform $cy(L)$ to a thin cylindrical 3-sized stack $ntp(L)$, as shown in Figure~\ref{fig:icosa_i-o} (d).

\begin{figure}[htbp]\centering
\includegraphics[width=\arxiv{.75}\columnwidth]{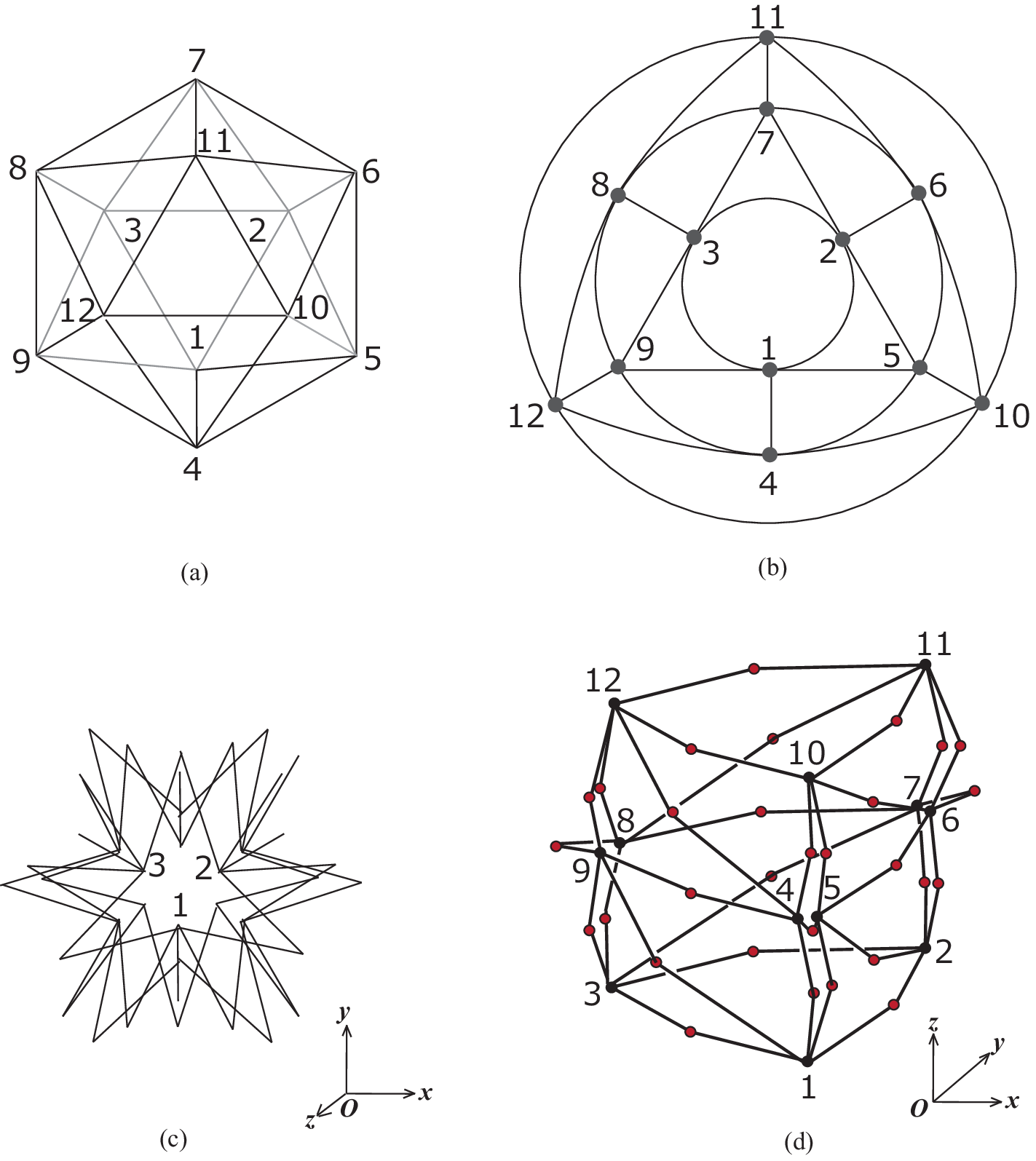}
\caption{(a) The linkage of a regular icosahedron in the 3-dimensional space; (b) The edge graph of a regular icosahedron (in general some circular edges can be missing); (c) A near spiky ball $nsb(L)$: (d) A near triangular prism linkage $ntpl (L)$. }
\label{fig:icosa_i-o}
\end{figure}

\begin{figure}[htbp]
\centering
\includegraphics[width=\arxiv{.75}\columnwidth]{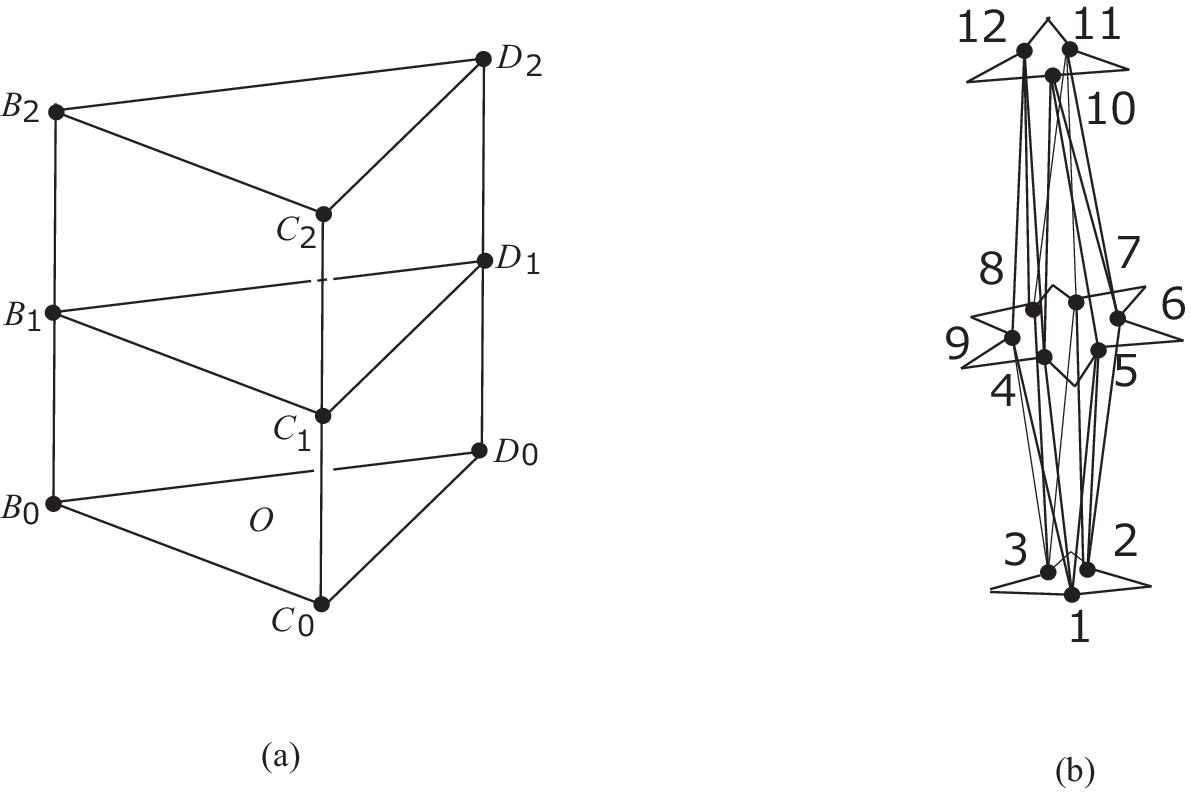}
\caption{(a) A triangular prism; (b) A 2-sized stack linkage $2-stack(L)$ of $L$. }
\label{fig:icosa_i-o_1}
\end{figure}

{\bf Step 5.} Subdivide each vertex set $V_i $ for $i=1, 2$, and $3$ into three groups
$\{V_{i,j}: j=1,2,3 \}$ evenly. Then, we continuously transform the $cy(L)$ to a near
 triangular prism linkage, denoted by $ntpl(L)$ in Figure~\ref{fig:icosa_i-o} (d). Note that $ntpl(L)$ contains a triangular prism as shown in Figure~\ref{fig:icosa_i-o_1} (a) used for passing and turning the linkage inside-out.

{\bf Step 6.} Fix the lower part of $ntp(L)$ and transform the upper part to make a modified tetrahedral linkage with folded edges in its exterior (see Figures \ref{fig:icosa_i-o_2}) (a) and (b). Then, transform the modified tetrahedral linkage inside-out by a similar motion as shown in Figure~\ref{fig:icosa_i-o_2} (c).

{\bf Steps 7--9.} By following similar processes used for the triangular prism linkage, we can continuously turn $L$ inside-out, as shown in Figure \ref{fig:icosa_i-o_2}, so we omit details.

\begin{figure}[htbp]
\centering
\includegraphics[width=\arxiv{.75}\columnwidth]{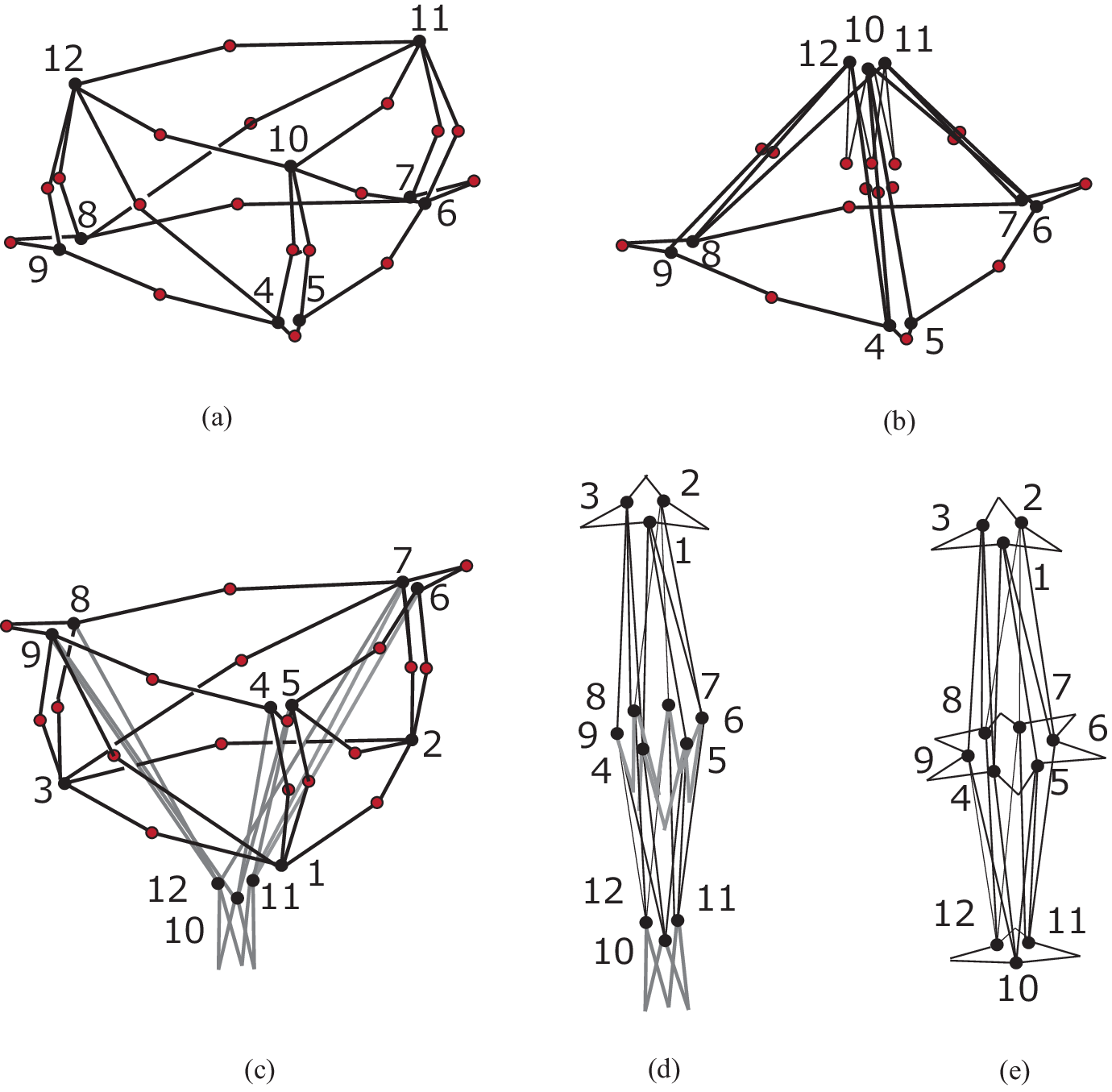}
\caption{(a) The upper part of the linkage; (b) Transform the upper part to a linkage of a modified triangular pyramid (tetrahedron) with folded edges; (c) After the operation of inside-out for the upper part; (d) After the operation of inside-out for the lower part; (e) The mirror image of Figure~\ref{fig:icosa_i-o} (a), obtained by moving horizontal edges
toward the exterior of the triangular cylinder.}
\label{fig:icosa_i-o_2}
\end{figure}

\subsection{Proof of Theorem~\ref{th:i-o}}\label{pr:i-o}
\noindent We assume the edge length of a given convex equilateral linkage $L$ is one from hereafter for the sake of simplicity.

\smallskip
{\bf Step 1.} Choose any face $F$ of $L$ and any interior point $O$ of $F$. Depict $L$ in the Euclidean space $\mathbb{R}^3$ with $xyz$-axes and the origin $O$ such that $F$ is on the $xy$-plane and $L$ are in the half space $\{z \ge 0\}$.

{\bf Step 2.}
Divide the vertex set $V$ into nonempty subsets $\{V_i: 0 \le i \le {\ell}\}$ according to the edge distance from $F$.
We call an edge whose both endpoints are in the same subset $V_i$ for some $i \,(0 \le i \le \ell)$ {\it horizontal edge} and the remaining edges  {\it vertical edges}, that is,  each of them has endpoints one in $V_i$ and the other in $V_{i+1}$ for some $0 \le i \le {\ell}-1$.

{\bf Step 3.}
Transform $L$ into a spiky ball by a continuous motion $\{L_t: 0 \le t \le 1\}$ defined in the proof of Theorem \ref{th:convex}. Stop the motion just before the end, e.g., $t=1-\epsilon$ for small enough $\epsilon >0$. We call the resulting linkage $L_{1-\epsilon}$ a {\it near-spiky ball}, denoted by {\it nsb(L)}.

{\bf Step 4.}
Subdivide each $V_l \, (0\le i \le {\ell})$ into three groups $V_{i,1}$, $V_{i,2}$, and $V_{i,3}$ evenly in equator order under the following condition, where $V_{\ell,1}=V_{\ell}$, and $V_{l,2}$ and $V_{\ell,3}$ are empty in the case $|V_{\ell}| \le 2$.
Consider the graph $G$ of the linkage $L$ and an extended graph $G'$ of $G$, which means $G$ is an induced subgraph of $G'$, such that

\begin{enumerate}
\item if  $|V_{\ell}| \ge 3$,  $G'$ has two more vertices $u$ and $v$ joining to all vertices in $V_0$ and $V_{\ell}$, respectively and

\item
if  $|V_{\ell}| \le 2$, $G'$ has one more vertex $u$ joining to all vertices in $V_0$.
\end{enumerate}

Then $G'$ is 3-connected by $|V_0| \ge 3$, and hence there are three internal-vertex-disjoint paths $P_j\, (j=1,2,3)$ joining any pair of vertices in $G'$ by Menger's Theorem \cite{Menger27,Goring00}. Choose $\{u, v\}$ as a pair of vertices where $v$ is any vertex in $V_l$ in the case of $|V_{\ell}| \le 2$.

Subdivide $V_i$ into $V_{i,j} (j=1,2,3)$ such that
 $$V_i \cap V(P_j) \subseteq V_{i,j},$$
where $V(P_j)$ means the vertex set in $P_j$, which is possible because $L$ is a convex polyhedral linkage.

By a small perturbation, move the vertices onto the surface of the cylinder around $z$-axis with radius $\epsilon$ such that for each $0 \le i \le \ell$  the vertices in $V_i$ are in the plane $\{z=\epsilon i /2 \} $. Simultaneously, the horizontal edges with endpoints in $V_i$ are moved together with $V_i$ targeting the positions in the plane $\{z=\epsilon i /2 \} $, and the midpoint of each vertical edge joining two vertices in $V_i$ and $V_{i+1}$ is moved to the position in the plane $\{z=\epsilon (i + 1/2)\} $ outside of the cylinder. The above motions can be continuous because $L$ is a convex equilateral polyhedral linkage.

{\bf Step 5.} Move vertices of the perturbed {\it nsb(L)} by parallel transformations along $z$-axis to a thin cylindrical $(\ell+1)$-sized stack, such that the vertices of $V_i$ are on the plane $z= i /3$ and edges are moved following their endpoints, where their midpoints are located in the exterior of the cylinder and horizontal edges are located on the planes orthogonal to $z$-axis. The motion can be continuous without collision because all edges folded in half are adjustable. We call the resulting figure a {\it cylindrical linkage}, denoted by $cy(L)$.

{\bf Step 6.}  To transform $cy(L)$ to a near triangular prism with extra edges, move vertices in $V_{i,1} \, (0 \le i \le \ell)$ to points in the circular arc $\Gamma_{i,1}$  obtained as the intersection of the circle of center $(0, 0, i/3)$ with radius $1/2 +\epsilon$ in the plane $z= i /3$,  and the sphere of center  $(1/2+\epsilon, 0, i /3)$ with radius $\epsilon$, where vertices keep the equator order. Similarly, Move vertices in $V_{i,2}$ and $V_{i,3}$ to points in the circular arc $\Gamma_{i,2}$ and $\Gamma_{i,3}$ obtained as rotated $\Gamma_{i,1}$ with the angles $\pm 2/3 \pi$ at $(0, 0, i /3)$ in the plane $z= i /3$. Then the distance of endpoints of each vertical edge is less than one because of $\{(\sqrt{3}(1/2 +\epsilon)\}^2 + (1/3)^2 < 1$ for $\epsilon$ very small. We call the resulting figure of $L$ a {\em near-triangular-prism linkage}, denoted by $ntp(L)$.

{\bf Step 7.} Transform the top part of $ntp(L)$ to a modified tetrahedral linkage (see Figure~\ref{fig:icosa_i-o_2} (b)). Apply Lemma \ref{lem:pyramid} such that horizontal edges of the top stack are rotated inside before the operation and then go down together with endpoints.

{\bf Step 8.} Continue the above process until all parts are turned inside-out. The result is a very thin $cy(L)$. Move all horizontal edges to the exterior of the cylinder.

{\bf Step 9.} Transform the resulting figure to nsb(L), and then transform it to a convex polyhedron, which is a mirror image of~$L$.

\section{Reversing Nonequilateral Linkages} \label{sc:i-o-noneq}
\noindent
What happens in the case of \emph{nonequilateral} convex polyhedral linkages?
We initialize research into this question by solving the asymptotic number
of required subdivisions for tetrahedra.
First we prove a general lower bound:

\smallskip
\begin{theorem}
To turn a polyhedron $P$ inside-out through a face $f$, each edge $e$ of $P$'s skeleton must be divided at least
$\Omega\left(\log \frac{\operatorname{length}(e)}{\operatorname{dist}(e,f)+\operatorname{circumf}(f)/2}\right)$.
\end{theorem}
\begin{proof}
This proof is based on the ``knitting needles'' argument
\cite[Section~6.3.1]{demaine2007geometric}.
Refer to Figure~\ref{fig:noneqlb}.
Let $p$ be the shortest path from a vertex of $f$ to a vertex $u$ of $e$,
which has length $\operatorname{dist}(e,f)$.
Draw a sphere $S_0$ centered at any vertex of $f$ of radius
$r_0 = \operatorname{dist}(e,f) + \operatorname{circumf}(f)/2$.
Then the path $p$ must remain inside the sphere $S$, given its short length.
If the next subdivided chunk of edge $e$ has length $> 2 r_0$,
then the other endpoint will always be outside $S$, which means this edge can never pierce $f$ as required.
Thus the next subdivided chunk has length $\ell_1 \leq 2 r_0$.
Similarly, the next subdivided chunk has length $\ell_2 \leq 2 (r_0 + \ell_1)$;
then $\ell_3 \leq 2 (r_0 + \ell_1 + \ell_2)$; and so on.
The $\ell_i$ upper bounds increase exponentially,
so we need $\Omega\left(\log \frac{\operatorname{length}(e)}{r_0}\right)$
subdivisions before we finish the edge.
\end{proof}

\begin{figure}[h]
    \centering
    \includegraphics[width=\arxiv{.75}\columnwidth]{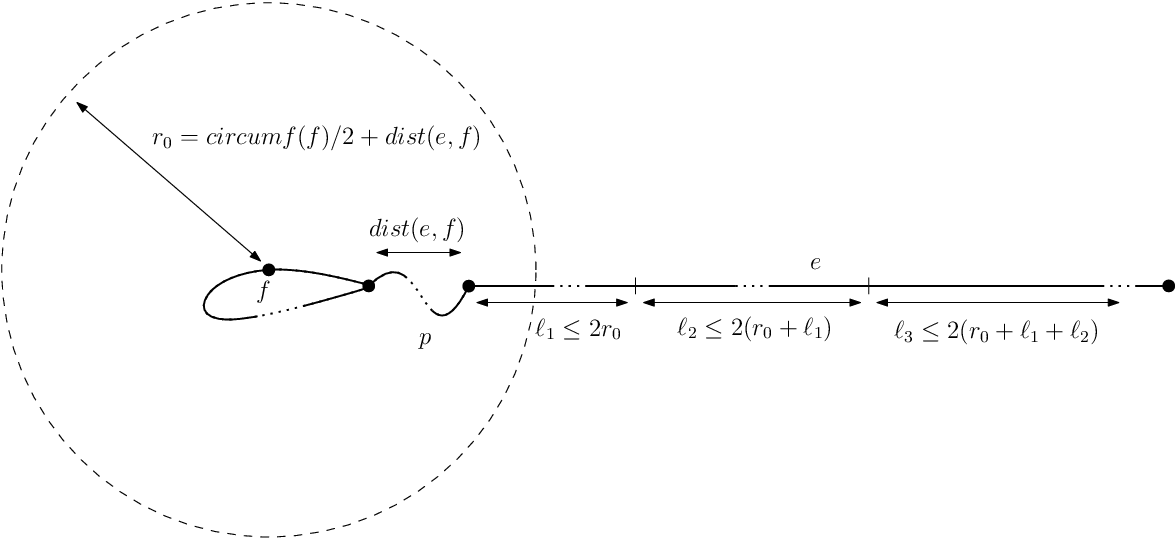}
    \caption{Visualization of largest lower bound constraints for an edge~$e$ through a face~$f$.}
    \label{fig:noneqlb}
\end{figure}

Now we prove a matching upper bound for tetrahedra, up to constant factors:

\smallskip
\begin{theorem}
A tetrahedron with edge lengths between $1$ and $L$ can be turned inside-out
through a specified face~$f$ using only $O(\log L)$ subdivisions.
\end{theorem}
\begin{proof}
    We focus on the worst case where $f$ is a triangle with edge lengths $1$ and $L\gg 1$.

    We start with the easier case where
    instead of three triangles connected to $f$, there is a single segment of length $L$ we want to push through $f$; see Figure~\ref{fig:nonequilateral}(a).
    The first subdivision is after length $\ell_1 = 1$, ensuring the segment fits through $f$. The next subdivision is after length $\ell_2 = \ell_1 + 1$, the next ones are after $\ell_3 = \ell_1 + \ell_2 + 1$, $\ell_4 = \ell_1 + \ell_2 + \ell_3 + 1$, and so on. The length $\ell_i$ grows exponentially in $i$, which is easy to observe from %
    the Fibonacci recurrence. Consequently, the presented approach uses $O(\log L)$ subdivisions.

    Figure~\ref{fig:nonequilateral}(b) shows the same technique with triangles instead of a single segment. The method works equivalently, but we need more space to push the triangles down instead of $\ell_1=1$ and $+1$ in the above computation. Indeed, we push all three triangles into the center position, resulting in $\ell_1=1/\sqrt(3)$ and $+1/\sqrt(3)$ (circumradius of size-$1$ triangle). This concludes the proof for the tetrahedron.
\end{proof}

\begin{figure}
    \centering
\includegraphics[width=\arxiv{.75}\columnwidth]{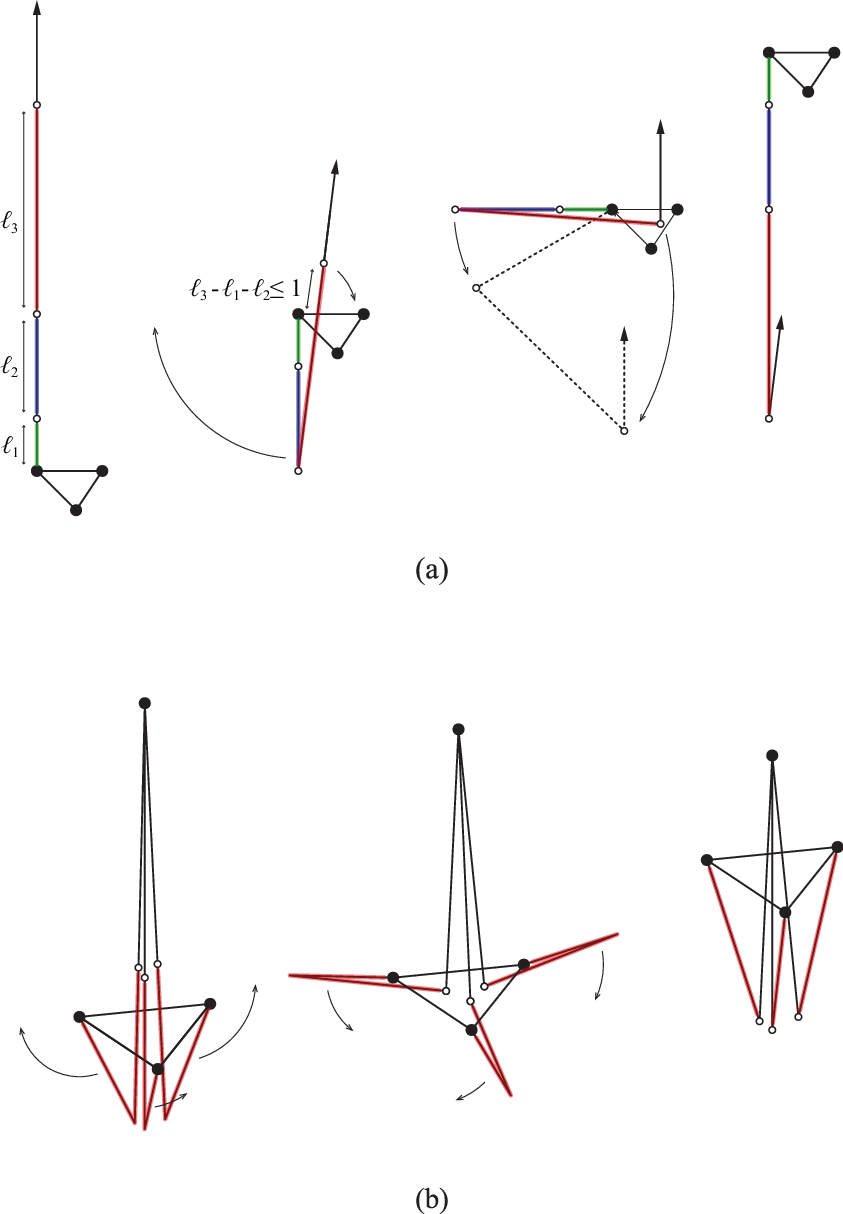}
    \caption{(a): Logarithmic nonequilateral algorithm for turning a single segment on top of an equilateral triangle inside-out. (b): Logarithmic nonequilateral algorithm for turning a tetrahedron inside-out by pushing through the equilateral base triangle of size $1$.}
    \label{fig:nonequilateral}
\end{figure}

\section{Open Problems}\label{sec:open}

\noindent In addition to the previous section (nonequilateral), one
could even further generalize to nonconvex polyhedra or
arbitrary linkages.
A finite upper bound (but lots of subdivision) is probably possible.
For flattening, does midpoint subdivision always suffice,
or $O(n)$ subdivision,
or is there a counterexample where more subdivision is needed?
Is it NP-hard or even $\exists \mathbb R$-hard to decide whether
a given linkage can be continuously flattened or turned inside-out?

\bibliography{bib}

\journal{
\begin{biography}
\profile[photos/edemaine]{Erik D. Demaine}{received a B.Sc. degree from Dalhousie University in 1995,
and M.Math. and Ph.D. degrees from University of Waterloo in 1996 and 2001, respectively. Since 2001, he has been a
professor in computer science at the Massachusetts Institute of Technology. His
research interests range throughout algorithms, from data structures for improving web searches to the geometry of understanding how proteins fold to the computational
difficulty of playing games. In 2003, he received a MacArthur Fellowship as a ``computational geometer tackling and solving difficult problems related to folding and bending—moving readily
between the theoretical and the playful, with a keen eye to revealing the former in the latter''. He cowrote a book about the theory of folding, together with Joseph O’Rourke (\textit{Geometric Folding
Algorithms}, 2007), and a book about the computational complexity of games, together with Robert Hearn (\textit{Games, Puzzles, and
Computation}, 2009).}
\profile[photos/mdemaine]{Martin L. Demaine}{is an artist and mathematician. He started the first private
hot glass studio in Canada and has been
called the father of Canadian glass. Since
2005, he has been the Angelika and Barton Weller Artist-in-Residence at the Massachusetts Institute of Technology. Both
Martin and Erik work together in paper,
glass, and other material. They use their exploration in sculpture
to help visualize and understand unsolved problems in mathematics, and their scientific abilities to inspire new art forms. Their
artistic work includes curved origami sculptures in the permanent
collections of the Museum of Modern Art (MoMA) in New York,
and the Renwick Gallery in the Smithsonian. Their scientific
work includes over 100 published joint papers, including several
about combining mathematics and art.}
\profile[photos/hecher_photo]{Markus Hecher}{is a postdoctoral researcher at MIT CSAIL working with Erik Demaine. He received a binational PhD in computer science from the Vienna University of Technology (Austria) and the University of Potsdam (Germany). His main interests are counting problems, computational complexity, and utilizing new insights into the hardness of combinatorial problems to improve existing algorithms.}
\profile[photos/rebecca]{Rebecca Lin}{is a Ph.D. student at MIT CSAIL working with Erik Demaine. She received her B.Sc. in Computer Science from the University of British Columbia, where she was advised by William Evans. Her research explores geometrical problems in art, design, and fabrication.}%
\profile[photos/vluo]{Victor Luo}{graduated MIT Class of 2023 with a B.S. in Mathematics and EECS, and received an M.Eng. under Erik Demaine in 2024. %
He has interned with Webstaff Co., Ltd. in Shibuya, and currently works as a software engineer at Doppel.
In his spare time, he enjoys playing rhythm games and square dancing.
}
\profile[photos/selfphto_nara]{Chie Nara}{received her B.A., M.S., and Ph.D. degrees from Ochanomizu University in Tokyo.
She served as a professor at Tokai University before taking up the current research position at Meiji University. Her
research fields are functional analysis, graph theory, and discrete geometry.}
\end{biography}}

\end{document}